\documentclass[compsoc,conference,a4paper,10pt,times]{IEEEtran}
\IEEEoverridecommandlockouts

\usepackage{cite}
\usepackage{amsmath,amssymb,amsfonts,amsthm}
\usepackage{algorithmic}
\usepackage{mathtools}
\usepackage{graphicx}
\usepackage{textcomp}
\usepackage{bmpsize}
\usepackage{xcolor}
\usepackage{lipsum}
\usepackage[colorlinks=true,urlcolor=black]{hyperref}
\def\BibTeX{{\rm B\kern-.05em{\sc i\kern-.025em b}\kern-.08em
    T\kern-.1667em\lower.7ex\hbox{E}\kern-.125emX}}

\input bksymbol2.sty

\newcommand{\KURZ}[2]{#1}
 
\newcommand{\FormelKurz}[2]{#1}
\newcommand{\ALT}[1]{}
\newcommand{\FNS}[1]{}

\newcommand{\Rblue}[1]{\textcolor{black}{#1}}

\newcommand{\Rteal}[1]{\textcolor{black}{#1}}

\newcommand{\N}{\mathbb{N}}

\newcommand{\E}{\mathbb{E}}

\newcommand{\T}{\mathcal{T}}

\newcommand{\fra}{\mathfrak{a}}
\newcommand{\frC}{\mathfrak{C}}

\newcommand{\Max}[1]{\underset{#1}{\textrm{max}}\  }

\newcommand{\DX}[1]{\ \mathbf{d} #1}

\newcommand{\MSPP}{{\tt SPC}}

\newcommand{\DRP}{random partition}%{\tt RP}}
\newcommand{\SAMP}{{\tt SAMP}}

\renewcommand{\ee}{\varepsilon}

\newcommand{\mujw}{\mu_{j,w}}
\newcommand{\mujv}{\mu_{j,v}}
\newcommand{\Djw}{\caD_{j,w}}
\newcommand{\Djv}{\caD_{j,v}}
\newcommand{\omuj}{\ovl{\mu_j}}

\newcommand{\hmu}{\hat{\mu}}

\newcommand{\hDD}{\hat{\DD}}

\newcommand{\Cjk}{\frC_{j \mapsto k}}

\newcommand{\Tjk}{\T_{j \mapsto k}}

\newcommand{\Tns}{\T[n,s]}

\newtheorem{theorem}{Theorem}[section]

\newtheorem{lemma}{Lemma}[section]
\newtheorem{definition}{Definition}[section]
\newtheorem{example}{Example}

\begin{document}

\title{Parallel Composition for Statistical Privacy
\thanks{ This research has been conducted within the AnoMed project (https://anomed.de/)
   funded by the BMFTR  (Bundesministerium für Forschung, Technologie und Raumfahrt)
  and the European Union in the NextGenerationEU action.}
}
\author{\IEEEauthorblockN{Dennis Breutigam, R\"udiger Reischuk \\}
\IEEEauthorblockA{\textit{Institute for Theoretical Computer Science, 
Universit\"at zu L\"ubeck, Germany}\\
d.breutigam@uni-luebeck.de, ruediger.reischuk@uni-luebeck.de}
% \and
% \IEEEauthorblockN{R\"udiger Reischuk}
% \IEEEauthorblockA{\textit{Institute for Theoretical Computer Science} \\
% \textit{Universit\"at zu L\"ubeck}\\
% L\"ubeck, Germany \\
% ruediger.reischuk@uni-luebeck.de}
}

\maketitle

\begin{abstract}
    Differential Privacy (DP) considers a scenario in which an adversary has almost
    complete information about the entries of a database. This worst-case assumption
    is likely to overestimate the privacy threat faced by an individual in practice.
    In contrast, Statistical Privacy (SP), as well as related notions such as
    noiseless privacy or limited background knowledge privacy, describe a
    setting in which the adversary knows the distribution of the database entries,
    but not their exact realizations. In this case, privacy analysis must account
    for the interaction between uncertainty induced by the entropy of the underlying
    distributions and privacy mechanisms that distort query answers, which can be
    highly non-trivial.

    This paper investigates this problem for multiple queries (composition). A
    privacy mechanism is proposed that is based on subsampling and randomly
    partitioning the database to bound the dependency among queries. This way
    for the first time, to the best of our knowledge, upper privacy bounds
    against limited adversaries are obtained without any further restriction on
    the database. 
    
    These bounds show that in realistic application scenarios taking
the entropy of distributions into account yields improvements of privacy and precision guarantees. 
We illustrate examples where for fixed privacy parameters and utility loss
SP allows significantly more queries than DP.

 \end{abstract}

\begin{IEEEkeywords}
privacy, accuracy, sampling, composition
\end{IEEEkeywords}

\section{Introduction}
Privacy preserving guarantees based on the concept of  differential privacy
often overestimate the knowledge of an adversary. Such an assumption requires
large  distortion of the answers to queries and hence leads to a significant
utility loss \cite{JD22,D21}. The situation that the adversary knows the entire
database except  a single critical entry is much too pessimistic in many real
world scenarios. In practice, the adversary's prior knowledge - also known as
background knowledge - is typically limited. Notions like noiseless privacy,
(inference-based) distributional differential privacy, and passive partial
knowledge differential privacy first introduced by Bhaskar et. al. and later
improved by different authors \cite{BBG11,BGKS13,DMK20} attempt to capture more realistic
scenarios.

In this paper we analyze the situation where only the distribution of the
database entries is publicly known, but not their exact values nor any further
information, called \tit{statistical privacy} (SP) in \cite{BR25}. This setting,
when adding noise to increase privacy further, is difficult to analyze due to
the complexity of the convolutions of the different distributions involved. 

To understand the privacy leakage of a single query  is not sufficient 
since often databases are subject to a bunch of queries --
in the privacy literature called composition.
Now, because of possible dependencies things get significantly more complicated,
in particular if the entries are not identically distributed.
In the DP setting when using independent noise for the distortion of each answer
the analysis remains manageable -- the leakages adds which was first examined by
Dwork et. al. \
\cite{Dwo+06} and  later improved by Kairuz et. al., Abadi et.
al. and Doroshenko et. al. \cite{KOV15,ACG16, DGKKM22 }.
For SP, however, each additional answer may reveal information about the
actual values of the noncritical entries, thus lowers the uncertainty of an
adversary. This holds even more if the correct answer is given every time, in
other words no additional entropy is introduced as it is in the DP setting.

Furthermore, if an adversary may generate queries adaptively the situation becomes
very complex.
We apply the subsampling technique provided in \cite{BR25_2} to the composition of queries based on a
random partition of the entries and prove upper bounds for privacy and utility loss.

\Rblue{
Our main results are the following:
\begin{itemize}
    \item the design of a framework for formally analyzing partition
    techniques with respect to privacy estimation,
    \item 
     general composition bounds against adversaries with limited background knowledge 
    without any requirements like injecting new entropy before each query, 
    %(section \ref{section:composition}),
    \item  
    construction of application scenarios showing that the entropy of database distributions
    as taking into account by SP can allow more composition queries than in the DP setting.
  % (section \ref{section:applications}).
%    \item techniques for incorporating structural knowledge that an adversary
%    may possess about a data source.
\end{itemize}
}

This paper is organized as follows. In the next section the problem setting is
defined formally. In section \ref{section:partitioning} we develop the framework
for partitioning mechanisms. Furthermore, extensions to more background
knowledge are considered.

Section~\ref{section:composition} investigates composition of queries. 
We conclude with a short resume and further research questions for statistical privacy.
An appendix provides more details and an introduction to the basics 
of stochastic process theory used in the proofs.\\

\subsection{Related Results}
Distribution-based privacy analyses were introduced by Bhaskar et al. in
\cite{BBG11}. However for composition, they made the restriction that the database changes
substantially between consecutive queries to generate new uncertainty --
comparable to the DP case where each answer is enriched with new independent
noise. This paper does not make any restrictions -- the database distribution
does not change for the whole sequence of queries. 

Distribution-based privacy was extended by Bassily et. al. in \cite{BGKS13} to align with the modern
indistinguishability-based formulation of DP, and Desfontaines et. al.
\cite{DMK20} further generalized the framework by introducing adversaries with
varying types of background knowledge as well as methods for accounting for
dependencies between queries. In \cite{BR25_2} we have applied subsampling
techniques to this setting.  

% \Rblue{For classical DP, privacy amplification through subsampling was first studied by
% Bassily et. al.\ in \cite{BST14} and was later refined by Balle et. al.\cite{BBG18}, who provided a detailed analytical understanding of how
% sampling interacts with noise-adding mechanisms. Building on this foundation, a
% series of subsequent works by Balle et. al. \ \cite{WBK19, BBG20, IC21} employed sampling to
% amplify the privacy of composed queries. The use of database
% partitioning--extended to distributed settings via random allocation--as a
% privacy--enhancing mechanism for composition was introduced by Bittau et. al.\ in
% \cite{BEM17} and was further developed in \cite{BKM20, FS25}.}

% \Rteal{While there are many articles on composed queries in DP, see among others}
% \cite{ACG16, M17, WBK19, DRS22}
% \Rteal{, the analysis of the case where the adversaries knowledge about
% the database contents are limited.  In }
% \cite{BBG11} 
% \Rteal{ bounds are established
% for databases that change significantly before the next query is asked.  In }
% \cite{DMK20}\Rteal{ a framework is introduced to take the dependencies between
% different queries into account. 
% In this paper we consider the case
% where we have multiple queries and a database of limited size and what to
% maximize the privacy gained under the assumption that an attacker only knows the
% distribution from which the database is drawn.
% Relevante Arbeiten: }
% \cite{FS25}

\section{Databases, Queries and Privacy}\label{section:statistical_privacy}
A database \wfinw{$D$} of size $n$ is a sequence of $n$ independent entries
$I_1, \ldots, I_n$, where each entry $I_j$ takes values in a common space $W$.
The random variable $I_j$ is distributed according to a marginal distribution
$\mu_j$ on $W$. These marginal distributions need not be identical; however,
their supports are assumed to coincide and are therefore equal to $W$.

The distributions $\omuj$ jointly induce a distribution (respectively, a density
function) $\mu$ on $W^n$, and we write $\caD \sim \mu$ for the corresponding
random variable on $W^n$.\footnote{
In the following, we do not distinguish between a distribution $\mu$ and a random
variable $\caD$ with distribution $\mu$, and use both terms
interchangeably when the meaning is clear from the context.}

For $w \in W$, we denote by $\mujw$ and $\Djw$ the conditional
distribution and the corresponding conditional random variable obtained by
fixing the entry $I_j$ to the value $w$.

Let $\caF$ be a set of queries that may be posed to a given database.
Formally, a query is modeled as a measurable function
\[
F : W^n \to A,
\]
where $A$ denotes the set of possible answers.
We assume that both $W$ and $A$ are totally ordered.

A query $F$ is called \emph{monotone} if
\[
D \le \hat{D} \;\Longrightarrow\; F(D) \le F(\hat{D}),
\]
where the partial order on $W^n$ is defined componentwise, that is,
$D = (I_1,\ldots,I_n) \le \hat{D} = (\hat{I}_1,\ldots,\hat{I}_n)$ if
$I_j \le \hat{I}_j$ for all $j \in \{1,\ldots,n\}$.

A privacy mechanism $\caM$ is a function that maps database queries to
random variables on $A$. That is, for a given database $D$ and a query $F$,
$\caM(F,D)$ is a random variable taking values in $A$ that perturbs the
true answer $F(D)$. If the database itself is modeled as a random variable
$\caD$, then $\caM(F,\caD)$ involves two sources of
randomness: the intrinsic uncertainty of the database $\caD$ and the
external noise introduced by the mechanism $\caM$.

To measure the privacy loss induced by a query $F$, one focuses on a
\emph{critical entry} $j$ and assesses how much information about its value
$I_j$ can be inferred from the query answer. Rather than attempting full
reconstruction of $I_j$, privacy analyses typically consider the simpler task of
distinguishing between two neighboring databases that differ only in whether
$I_j = v$ or $I_j = w$ for $v,w \in W$, while all other entries remain fixed.

To protect privacy one should not allow queries that treat entries differently
-- in the extreme case depend only on a small subset of entries. Then obviously
it is quite easy to deduce information about these few entries. Sampling only
makes sense if the random order in which elements are drawn does not matter.
For these reasons we restrict the analysis to symmetric functions $F$.
Furthermore, applying  sampling techniques  to determine, resp.~approximate a property of 
the whole database makes only sense for queries and distributions
for which the variance of the answer for a random sample is bounded and from this answer
one can deduce the answer for the whole database.
We will not further elaborate on these statistical issues here,
but note that for most queries that have been considered so far in the privacy literature
-- most notably  summation type queries -- and \emph{well behaved} distributions
subsampling can be applied.

Now the question arises how to measure a potential leakage of private
information when the answer to a database query is released. Differential
privacy assumes that an adversary has complete knowledge about the values of all
other noncritical entries which seems to be the worst situation one can think of
\cite{Dwo06}.

\begin{definition}[Differential Privacy (DP)]
For  a collection   $\caW$ of databases a
privacy mechanism $M$ achieves $(\ee,\dd)$-differential privacy for a query $F$
if for all neighboring databases $D,D'$ in $\caW$ and  for all $S \subseteq A$ 
it holds
\[  \Prob{M(F,D) \in S}{} \ \kla  e^\ee \ \Prob{M(F,D') \in S}{}   + \delta  \ .   \hspace{2ex} \wkast \]

\end{definition}

In this setting the distribution of the noncritical entries is trivial -- fixed to unique values.
Analytically such a situation is easier to handle,
but this assumption about background knowledge
is far too strong in many practical situations.
If  good privacy protection  in such a
worst case scenario can be achieved one is on the safe side.
However, this typically requires $M$ to induce large distortion on the
correct answer for $F$, thus the precision of such an answer is low.

Better utility can be achieved for more realistic
scenarios with less background knowledge of an adversary.
Such a setting has been called \emph{noiseless privacy} \cite{BBG11} or
\emph{differential privacy with partial knowledge}  \cite{DMK20}.
We concentrate on the uncertainty generated by  database distributions
combined with privacy enhancing mechanisms like distortion by noise or subsampling 
and use the term \emph{statistical privacy} \cite{BR25}.
Now one has to determine the convolution of distributions 
for which no simple formulae are known in general.

\begin{definition}[Statistical Privacy (SP)]
A privacy mechanism $M$ achieves  $(\ee,\dd)$-statistical privacy for a critical entry $j$ 
%at $v,w  \in W$
with respect to a distribution $\mu \sim \caD$ 
(or a collection of distributions) and a query $F$ 
if (for all $\mu$) for every subset $S \sse A$  and all pairs $v,w \in W$ it  holds:
% for $\caD \sim \mu$:
\[  \Prob{M(F, \caD_{j,v}) \in S}{} \ \kla  
    e^\ee \ \Prob{M(F,\caD_{j,w}) \in S}{}   + \delta  \ .   \]
It achieves  $(\ee,\dd)$-statistical privacy if this condition holds for every
entry $j$.  \wkast

% \Rred{
% Given $\ee$ let  $\Phi^j_{\mu,F,M}(\ee)$  denote the smallest $\dd$ such that $M$ 
% for $F$ and $\mu$ achieves $(\ee,\dd)$-statistical privacy for $j$.
% }
% In case that no operator $M$ is used 
% %(noiseless privacy -- $M$ is the identity), 
% we simply write $\Phi_{\mu,F}$ --
% thus privacy is only generated by the distribution $\mu$.
\end{definition}

\Rteal{
In the following analysis we will consider only a single $j$. 
The bounds obtained can trivially be extended to all entries by taking only pairs
$(\ee,\dd)$ that hold for all $j$.  Note that for identically distributed entries the
conditional distributions $\caD_{j,v}$ are identical as well.
}

Given a (random) function $G: \; X \to A$, by $\caK_G$ we denote the Markov
kernel corresponding to $G$ that maps a distribution  $\mu$ on $X$ to the
distribution $\mu \; \caK_G$ on $A$ that is generated by $G$.

\begin{definition}%{\rm  { } \\
For $G: \; X \to A$ and $\mu,\hmu$  distributions on $X$, the \emph{privacy
curve} $\DD^G_{\mu, \hmu}$ is defined for $\ee \ge 0$  as follows  \cite{BBG18}:
\FormelKurz{ \begin{IEEEeqnarray*}{rCl}
\DD^G_{\mu, \hmu}(\ee) 
  &:=&  \int_A \max   \lk 0, \; \mu   \caK_G (a) - e^{\ee} \; \hmu \caK_G (a) \rk \DX{a} \\ 
   &\gla&   \max_{S \sse A} \  \int_S   \mu   \caK_G (a) - e^{\ee} \; \hmu
   \caK_G (a)  \DX{a}    \ . \hspace*{2ex} \wkast
 \end{IEEEeqnarray*}}{ \begin{IEEEeqnarray*}{rCl}
\DD^G_{\mu, \hmu}(\ee) 
  &:=&  \int_A \max   \lk 0, \; \mu   \caK_G (a) - e^{\ee} \; \hmu \caK_G (a) \rk \DX{a}             
   \gla   \max_{S \sse A} \  \int_S   \mu   \caK_G (a) - e^{\ee} \; \hmu \caK_G (a)  \DX{a}    \ .           \wkast
 \end{IEEEeqnarray*}}
\end{definition}

In the following we will only analyze the effect of subsampling opperators
\SAMP{} applied to a query $F$ to improve privacy. Besides subsampling other
privacy mechanisms $M$ may be used in addition. 
Our results will be independent of the specific type of $M$. To
simplify the notation we omit the additional term $M$ and use $F$ as a
combination of the query and an arbitrary $M$. Let 
\[\hDD^{F \circ \SAMP}_{\caD,j}(\ee) \dea \max_{v,w} \DD^{F \circ \SAMP}_{\caD_{j,v},
\caD_{j,w}}(\ee) .\]
Then a sampling operator $\SAMP$ achieves $(\ee,\dd)$-statistical privacy for $\caD$, $F$ and $j$ iff $\dd \ge
\hat{\DD}^{F \circ \SAMP}_{\caD,j}(\ee)$. 

\section{Partitions and Templates}\label{section:partitioning}
To specify a random selection of entries term template was defined in
\cite{BR25_2} this notion can be expanded to allow for drawing multiple samples
making a partition of the database. To now ask multiple queries on these smaller
databases we consider for a function $F$ with $n$ arguments a family of
functions $\caF = (F_s(y_1, \dots, y_s))_{s \in \N}$ that approximate $F$.Thus
in the following when we talk about a query $F$ we actually mean the family
$\caF$ and depending of the size of the sample the appropriate member. If a
distribution is not smooth, but there are very few entries that determine the
extrema it is unlikely that such an element is drawn and therefore one cannot
hope for an acceptable approximation. To analyze a composition of $m$ queries we specify $m$
subsamples simultaneously using the following definition.

\begin{definition}[Templates]\label{definition:template}
    Given a sequence $X=x_1,\ldots,x_n$ of fixed length $n$, a set of $m$
    samples $Y_k=y_{k,1},\ldots,y_{k,n_k}$ of $X$ for $k \in [1:m]$ can be
    described by $m$ sequences $\tau=\tau_1, \dots, \tau_m$ of indices such that
    $y_{l,k} = x_{\tau_k(l)}$. Such a $\tau$ is called a  \emph{template}
    and $\nu_\tau = (|\tau_1|,\dots, |\tau_m|)$  the template format. In the
    following, given a template $\tau$ let $n_k =|\tau_k|$. An \emph{injective
    template} is a template $\tau$ in which each index $i \in [1 \ldots n]$
    occurs at most once.

    Let $\frC{}$ denote the set of all  templates for a sequence of length $n$,
    $\frC{}(\nu)$ those of format $\nu$ for $\nu \in   \N^m$ and $\frC{}^k$  the set
    of all  subtemplates $\tau_k$ for the $k$-th sample. The subset of   $\frC{}$
    where the  $j$-th element of $X$ appears in the $k$-th sample is denoted by
    $\Cjk \dea \left\{ \tau \in \frC{} \mid j \in \tau_k \right\}$.    
    If only injective templates are considered we use the notation $\hat{\frC}$,
    $\hat{\frC}(\nu)$ etc.
    
    Given a template $\tau$, the operator $\SAMP_\tau: W^n \to W^{n_1} \times
    \dots \times W^{n_m}$ maps a database $\caD$ of size $n$ to $m$ databases
    $\caD_1, \ldots,\caD_m$   of  sizes $n_1,\dots, n_m$ and can be described
    by a Markov kernel denoted by $\caK_\tau$.
\wkast
\end{definition}

To define drawing a suitable partition we first consider the concept of a
sampling technique and Sampling Privacy Curve first introduced by \cite{BR25_2}.

\begin{definition}[Sampling Technique]
\label{definition:subsampling_technique}
    For fixed length $n$, a specific sampling technique is specified by a
    random variable $\caT \sim \mu_\caT$ on $\frC$.

    Combining \SAMP{} with a sampling technique $\caT$ we get an operator $\SAMP_\caT$ from $W^n$ to $W^\star$ with kernel
    \[ \caK_{\caT} \gla \int_{\frC} \caK_\tau  \DX{\mu_\caT(\tau)} \ . \]
\end{definition}

\begin{lemma}
  For a a database distribution $\mu$, a query   $F$, a sampling
  technique $\caT$ and $S \subseteq A$ it holds
  \[  \mu \; \caK_\T \; \caK_F (S) \gla  \sum_{\tau} \; \mu_\caT (\tau) \; \mu
  \; \caK_\tau \; \caK_F (S) \ .   \]
\end{lemma}

\begin{definition}[Sampling Privacy Curve (\MSPP)] \label{definition:mspp}
    For a database $\caD$, a query $F$ and a sampling technique 
    $\caT$ the sampling privacy curve of the $j$-the entry is  defined by
\begin{IEEEeqnarray*}{rCl}
  \MSPP^{F,\T}_{\caD,j}(\ee) 
%  \Rred{\qq \MSPP_{F,\caD,\T}^j(\ee)}
   &\dea & \E_{\tau \sim \T}
     \left[ \hDD^{F \circ \SAMP_\tau}_{\caD,j}(\ee) \right] 
        \ . \hspace*{2ex} \wkast
\end{IEEEeqnarray*}
\end{definition}

With respect to privacy, to generate large entropy by splitting a database into subdatabases
a good strategy is to choose a template at random to define a partition of the entries.

\begin{definition}[Random Partition]\label{definition:drp}
    Given a subset $\tilde{\frC}$ of $\frC$, a  $\tilde{\frC}$-random partition
    of a database $\caD$ is modeled by a  uniformly distributed random variable
    $\caR \sim \mu_\caR$ on $\tilde{\frC}$. Let $\caR^k$ denote the projection
    of $\caR$ onto the $k$-th sample and $\caR_{j,k} := \caT \mid \frC{}_{j,k}$
    the restriction that the $j$-th index belongs to the $k$-th sample.

    Combining  $\caR$ with the $\SAMP$ operator we get an operator $\SAMP_\caR$
    from $W^n$ to $W^{n_1} \times \dots \times W^{n_m}$ with kernel \(
    \caK_{\caR} \gla \int_{\tilde{\frC}} \caK_\tau  \DX{\mu_\caR(\tau)} \ . \)
    \wkast
\end{definition}

\section{Handling Background Knowledge}
The sample template framework can handle also scenarios where the adversary
possesses more background knowledge than the distribution. Now the entropy of
the samples might be quiet different. For example, he may know the exact values
of a subset of entries or structural information or that the entries fall into
different groups and the exact sizes of each group is given to him. This can be
modeled by representing known entries as point distributions or entry
distributions using different parameters. Now one faces a hypergeometric,
multivariate hypergeometric distribution or multinomial distribution generated
by the sampling process. 

\begin{example}\label{example:SPC_Bound_hypergeometric_dist_PQ}
    Let $\caD$ be a database of size $n$, where the adversary knows $v$ entries
    and the remaining entries are iid distributed. Let $F$ be a property query
    and let $\caT_{n,s}$ denote sampling without replacement producing a sample
    of size $s$. Furthermore, let $\caD^z_s$ be a database of size $s$ in which
    $z$ entries are known. The specific values of the known entries do not
    affect privacy -- since these values influence the distribution of
    $F(\cdot)$ independently of the sensitive entry. Only the number of known
    entries matters. Consequently, the sampling privacy curve can be computed as
    follows:
    \begin{IEEEeqnarray*}{rCl}
        \MSPP_{\mu, j}^{F, \Tns}(\ee)  &=& \sum^{s-1}_{z=0}
        \frac{\binom{v}{z} \binom{n-v}{s-1-z}}{\binom{n}{s-1}} \hat{\DD}^F_{\caD^z_s,j}(\ee).
    \end{IEEEeqnarray*}
    Furthermore, in cases where the privacy curve for a fixed $\ee$ drops
    rapidly at first and then flattens out -- as is the case for property
    queries (see Fig.~\ref{figure:PQ_Delta_Parameter_DB_Size}) -- the
    computational costs can be significantly reduced. Choosing a suitable
    threshold $\phi$ for $z$ one could replace the
    $\hat{\DD}^F_{\caD^z_s,j}(\ee)$ values by $1$ larger for $z > \phi$ and by
    $\hat{\DD}^F_{\caD^\phi_s,j}(\ee)$ for smaller $z$ because of monotonicity.
    This means for $0 \leq z < s-1$ 
    \begin{IEEEeqnarray*}{l}
        \MSPP_{\mu, j}^{F, \Tns}(\ee) \leq  \\
        (1-\textrm{cdf}_{\mathcal{H}(n,v,s-1)}(z))
        + \textrm{cdf}_{\mathcal{H}(n,v,s-1)}(z) \hat{\DD}^F_{\caD^z_s,j}(\ee) \ ,
    \end{IEEEeqnarray*}
    where $\mathrm{cdf}_{\mathcal{H}(n,v,s-1)}$ denote the cumulative distribution
    function of the hypergeometric distribution with parameters $n, v, s-1$. 
\end{example}

\section{Composition}\label{section:composition}

A database is often subject to several queries, called composition in the
privacy literature. Then the original uncertainty of an adversary will decrease
by every answer. To prevent this Bhaskar et. al.~have considered a situation
where the database changes significantly between each query by modifying or
adding a certain percentage of new data \cite{BBG11}, thus new entropy is
generated. But such a situation is not typical for many applications. For a
fixed database one could apply sampling to answer each query only with parts of
the database to limit the dependencies of the answers. Sampling reduces the
precision of answers which can be well estimated by statistical methods, but the
lower dependency of the answers might requires less distortion for each answer.
We will show that a good balance between this precision loss and the lower
utility loss by distortion achieves good privacy and good overall utility
simultaneously. 
%\Rred{Consider the Article }\cite{S10}

The lowest dependency between answers is obtained by using disjoint subdatabases
for every query. Furthermore, an adversary should have no information which
subdatabase contains the critical element. Thus the database curator chooses a
random assignment of the entries.

\Rteal{In the following we will use only injective templates. Then given a template
format $\nu$,   each entry has the same probability $n_k/n$ to end up in the
$k$-th subdatabase used for the $k$-th query. }
%curve in \cite{BR25_2}.

%\Rred{Composition Sampling verschoben zum allgemeinen Sampling Teil}

\subsection{Nonadaptive Composition}

Nonadaptive composition means that every query does not depend on answers to
other queries, thus they can be asked in parallel.
$m$ such queries can be modeled as a single query $F= (F_1, \dots, F_m):\ W^n \mapsto A^m$.
% with $F = (F_1, \dots, F_m)$  where $F_k: W^{n} \mapsto A$. 
If the curator computes the answers using a partition of the database
the format of the partition has to be fixed
in advance as well, let it be $\nu=(n_1,\ldots,n_m)$.
%By using a partition splitting the database in subsets of sizes $n_1, \dots, n_m$ 
%and using these to answer the partial queries $F_k$ they are
%approximated as a query on a database of size $n_k$ as discussed in the
%beginning of this section. Further in that case we understood the composed query
%as $F:W^{n_1} \times \dots \times W^{n_m} \mapsto A^m$ such that for $d = (d_1,
%\dots, d_m) \in W^{n_1} \times \dots \times W^{n_m}$ we have that $F(d) =
%(F_1(d_1), \dots, F_m(d_m))$. Under the assumption, that by some process the
%subset sizes are chosen and known by the curator and the adversary the following
%holds for the privacy of such a composition of queries.

When using a partition of the database to answer a sequence of queries an
adversary could try to estimate from the  answers which entries are used for
each query. 
If the entries are identically distributed this will not provide any information.
However, for non-identical distributions this could give some advantage.
%To still derive bounds for the privacy of a composition of queries
%these dependencies need to be taken into account. 
In this case the answer are no longer independent which complicates the analysis
significantly.
The sensitive entry can be member of any of the $m$ subsets that are used to 
answer the queries.

\begin{lemma}\label{lemma:nonadaptive_composition:pure_spc}
Let $\mu$ be a distribution on databases $\caD$ of size $n$, 
$F =(F_1, \dots, F_m)$ a nonadaptive composition of
$m$ queries \ and $\T$  denote the injective $(n_1, \dots, n_m)$--\DRP\ of the entries.
Then $F \circ SAMP_{\T}$  achieves $(\ee, \dd)$--statistical privacy for $\mu$ 
and critical element $j$ with
\[ \dd \gla   \sum_{k=1}^{m} \frac{n_k}{n} \; \
    \MSPP^{F_k, \caT_{j,k}}_{ \caD,j}(\ee)   \ . \]
\end{lemma}
\begin{IEEEproof} 
For functions $f$ and $g$ defined on the same space $A$, we may write $f(a) -
g(a)$ as $(f - g)(a)$ for notational clarity. Let  $v,w \in W$ be arbitrarily chosen and $\fra{} = (a_1, \dots,  a_m) \in A^m$ a vector 
of answers to the $m$ queries. 
%    Remember that $\caT_{j,k} = \caT \mid j \in \caT^k$.
    \begin{IEEEeqnarray}{l}
        \DD_{F,\; \mujv \caK_\T, \; \mujw \caK_\T}(\ee)   \nonumber \\
        = \Max{S \subseteq A^m} \int_{S} \mujv \caK_{\caT} \caK_F(\fra{}) - e^\ee \mujw \caK_{\caT} \caK_F(\fra{})  \DX{\fra{}}   \nonumber  \\
%        &\gla \Max{S \subseteq A^m} \int_{S} \ \sum_{k=1}^{m} \frac{n_k}{n}  \mujv \caK_{\caT_{j,k}} \caK_F(\fra{}) - \sum_{k=1}^{m} \frac{n_k}{n}  e^\ee \mujw \caK_{\caT_{j,k}} \caK_F(\fra{})   \DX{\fra{}}  \nonumber  \\
         \leq \sum_{k=1}^{m} \frac{n_k}{n} \ \Max{S \subseteq A^m} \nonumber \\
        \qd \qd   \int_{S} \mujv \caK_{\caT_{j,k}} \caK_F(\fra{}) -  e^\ee \mujw \caK_{\caT_{j,k}} \caK_F(\fra{})   \DX{\fra{}} \; .  \label{equation:nonadaptive:spc:1}
    \end{IEEEeqnarray}
Using    
\[          \mu \caK_{\caT_{j,k}} \caK_F(\fra{}) = \sum_{\tau \in \frC^k_j} \caT_{j,k}(\tau) \; \mu \caK_{\tau} \caK_F(\fra{})     \] 
    the integral in \eqref{equation:nonadaptive:spc:1} can be split.  
%    For the answers where the sensitive entry is not used ($k \neq j$) 
 The subqueries are independent for any fixed $\tau$
    due to the disjointedness of the subsets and the independence
    of the entries . 
    Therefore it holds 
    \[ 
        \mu \caK_{\tau} \caK_{F}(\fra{}) \gla \prod_{\ell=1}^{m} \mu \caK_{\tau^\ell} \caK_{F_\ell}(a_\ell). 
    \] 
    Since the sensitive element $j$ is not contained in the $l$-th
    subset for $\ell \neq k$  we have $\mujv \caK_{\tau^\ell}
    \caK_{F_\ell}(a_\ell) \gla \mujw \caK_{\tau^\ell} \caK_{F_\ell}(a_\ell) \gla
    \mu \caK_{\tau^\ell} \caK_{F_\ell}(a_\ell)$. Thus
    \KURZ{  \begin{IEEEeqnarray*}{rCl} 
        &&\mujv \caK_{\caT_{j,k}} \caK_F(\fra{}) -  e^\ee \mujw \caK_{\caT_{j,k}} \caK_F(\fra{}) \\
        &\gla& \sum_{\tau \in \frC_{j,k}} \caT_{j,k}(\tau)  \left(\mujv \caK_{\tau} \caK_F(\fra{}) -  e^\ee \mujw \caK_{\tau} \caK_F(\fra{}) \right) \\
%        &\gla& \sum_{\tau \in \frC_{j,k}} \caT_{j,k}(\tau) \left(  \mujv \caK_{\tau^k} \caK_{F_k}(a_k) \; \prod_{\ell \neq k} \mu \caK_{\tau^\ell} \caK_{F_\ell}(a_\ell)  -  e^\ee  \mujw \caK_{\tau^k} \caK_{F_k}(a_k) \; \prod_{\ell \neq k} \mu \caK_{\tau^\ell} \caK_{F_\ell}(a_\ell) \right) \\
        &\gla& \sum_{\tau \in \frC_{j,k}} \caT_{j,k}(\tau) \left(  \mujv
        \caK_{\tau^k} \caK_{F_k}  -  e^\ee  \mujw \caK_{\tau^k}
        \caK_{F_k} \right)(a_k) \\
        &&\hspace*{7ex} \cdot \prod_{\ell \neq k} \mu \caK_{\tau^\ell} \caK_{F_\ell}(a_\ell) 
     \end{IEEEeqnarray*}}{    \begin{align*} 
        &\mujv \caK_{\caT_{j,k}} \caK_F(\fra{}) -  e^\ee \mujw \caK_{\caT_{j,k}} \caK_F(\fra{}) \gla \sum_{\tau \in \frC_{j,k}} \caT_{j,k}(\tau)  \left(\mujv \caK_{\tau} \caK_F(\fra{}) -  e^\ee \mujw \caK_{\tau} \caK_F(\fra{}) \right) \\
        &\gla \sum_{\tau \in \frC_{j,k}} \caT_{j,k}(\tau) \left(  \mujv \caK_{\tau^k} \caK_{F_k}(a_k) \; \prod_{\ell \neq k} \mu \caK_{\tau^\ell} \caK_{F_\ell}(a_\ell)  -  e^\ee  \mujw \caK_{\tau^k} \caK_{F_k}(a_k) \; \prod_{\ell \neq k} \mu \caK_{\tau^\ell} \caK_{F_\ell}(a_\ell) \right) \\
        &\gla \sum_{\tau \in \frC_{j,k}} \caT_{j,k}(\tau) \left(  \mujv \caK_{\tau^k} \caK_{F_k}(a_k)  -  e^\ee  \mujw \caK_{\tau^k} \caK_{F_k}(a_k) \right)    \; \prod_{\ell \neq k} \mu \caK_{\tau^\ell} \caK_{F_\ell}(a_\ell) 
    \end{align*}}
Inserting this into \eqref{equation:nonadaptive:spc:1} gives 
    \KURZ{
    \begin{IEEEeqnarray}{l}
        \sum_{k=1}^{m} \frac{n_k}{n} \Max{S \subseteq A^m} \int_{S} \mujv \caK_{\caT_{j,k}} \caK_F(\fra{}) -  e^\ee \mujw \caK_{\caT_{j,k}} \caK_F(\fra{}) \DX{\fra{}} \nonumber \\
        = \sum_{k=1}^{m} \frac{n_k}{n} \ \Max{S \subseteq A^m} \int_{S}
        \sum_{\tau \in \frC_{j,k}}  \nonumber \\ 
        \hspace*{7ex} \caT_{j,k}(\tau) \cdot \big(  \mujv
        \caK_{\tau^k} \caK_{F_k}  -  e^\ee  \mujw \caK_{\tau^k}
        \caK_{F_k} \big)(a_k) \nonumber \\
        \hspace*{7ex} \prod_{\ell \neq k} \mu \caK_{\tau^\ell}
        \caK_{F_\ell}(a_\ell)  \DX{\fra{}} \nonumber \\
        \leq \sum_{k=1}^{m} \frac{n_k}{n} \sum_{\tau \in \frC_{j,k}}
        \caT_{j,k}(\tau) \Max{S \subseteq A^m} \nonumber \\
        \hspace*{7ex} \int_{S}  \left(  \mujv
        \caK_{\tau^k} \caK_{F_k}  -  e^\ee  \mujw \caK_{\tau^k}
        \caK_{F_k} \right)(a_k)  \label{equation:nonadaptive:spc:2} \\
        \hspace*{7ex} \prod_{\ell \neq k} \mu \caK_{\tau^\ell}
        \caK_{F_\ell}(a_\ell)  \DX{\fra{}} \nonumber
    \end{IEEEeqnarray}
    }{\begin{align}
        &\sum_{k=1}^{m} \frac{n_k}{n}\  \Max{S \subseteq A^m} \int_{S}  \mujv \caK_{\caT_{j,k}} \caK_F(\fra{}) -  e^\ee \mujw \caK_{\caT_{j,k}} \caK_F(\fra{}) \DX{\fra{}} \nonumber \\
        &\gla  \sum_{k=1}^{m} \frac{n_k}{n} \ \Max{S \subseteq A^m} \int_{S}  \sum_{\tau \in \frC_{j,k}} \caT_{j,k}(\tau) \left(  \mujv \caK_{\tau^k} \caK_{F_k}(a_k)  -  e^\ee  \mujw \caK_{\tau^k} \caK_{F_k}(a_k) \right)    \; \prod_{\ell \neq k} \mu \caK_{\tau^\ell} \caK_{F_\ell}(a_\ell)  
        \DX{\fra{}} \nonumber \\
        %THIS IS WHERE WE LOSE THE INFORMATION PROVIDED BY THE OTHER QUERIES
        &\kla \sum_{k=1}^{m} \frac{n_k}{n} \sum_{\tau \in \frC_{j,k}} \caT_{j,k}(\tau) \Max{S \subseteq A^m} \int_{S}  \left(  \mujv \caK_{\tau^k} \caK_{F_k}(a_k)  -  e^\ee  \mujw \caK_{\tau^k} \caK_{F_k}(a_k) \right)    \; \prod_{\ell \neq k} \mu \caK_{\tau^\ell} \caK_{F_\ell}(a_\ell)  
        \DX{\fra{}} 
        %\label{equation:nonadaptive:spc:2}
    \end{align}}
    Using Fubini's theorem to split the integral over $S = S_1 \times \dots
    \times S_m$ in $m$ integrals with respect to $a_1, \dots, a_m$. This allows
    to extract all densities independent of the sensitive entry and
    then to determine their integral using that the densities are positive
    functions and as such their maximizing integral is over the set $A$. 
\KURZ{
         \begin{IEEEeqnarray*}{rCl}
        &&\Max{S \subseteq A^m } \int_{S} \left( \mujv \caK_{\tau^k} 
    \caK_{F_k}  -  e^\ee  \mujw \caK_{\tau^k} \caK_{F_k} \right)(a_k) \\
        && \hspace*{10ex} \cdot \prod_{\ell \neq k} \mu 
        \caK_{\tau^\ell} \caK_{F_\ell}(a_\ell) \DX{\fra{}} \\
        &=& \Max{S \subseteq A^m } \int_{S_1} \dots \int_{S_m} \left( \mujv \caK_{\tau^k} \caK_{F_k}
        -  e^\ee  \mujw \caK_{\tau^k} \caK_{F_k} \right)(a_k) \\
        && \hspace*{10ex} \cdot \prod_{\ell \neq k} \mu \caK_{\tau^\ell} 
        \caK_{F_\ell}(a_\ell) \DX{a_1} \dots \DX{a_m} \\
        &=& \Max{S \subseteq A^m } \int_{S_k} \left(  \mujv \caK_{\tau^k}
        \caK_{F_k} -  
        e^\ee  \mujw \caK_{\tau^k} \caK_{F_k}\right)(a_k)  \DX{a_k} \\
        && \hspace*{10ex} \cdot\prod_{\ell \neq k} 
        \underbrace{\int_A \mu \caK_{\tau^\ell} \caK_{F_\ell}(a_\ell) \DX{a_\ell} }_{=1} \\
        &=& \DD^{F_k \circ \SAMP_{\tau^k}}_{\Djv, \Djw}(\ee)
     \end{IEEEeqnarray*}}{    \begin{align*}
        &\Max{S \subseteq A^m } \int_{S} \left( \mujv \caK_{\tau^k} 
        \caK_{F_k}(a_k)  -  e^\ee  \mujw \caK_{\tau^k} \caK_{F_k}(a_k) \right) \; \prod_{\ell \neq k} \mu 
        \caK_{\tau^\ell} \caK_{F_\ell}(a_\ell) \DX{\fra{}} \\
        &\gla \Max{S \subseteq A^m } \int_{S_1} \dots \int_{S_m} \left( \mujv \caK_{\tau^k} \caK_{F_k}(a_k)  
        -  e^\ee  \mujw \caK_{\tau^k} \caK_{F_k}(a_k) \right) \; \prod_{\ell \neq k} \mu \caK_{\tau^\ell} 
        \caK_{F_\ell}(a_\ell) \DX{a_1} \dots \DX{a_m} \\
        &\gla \Max{S \subseteq A^m } \int_{S_k}  \mujv \caK_{\tau^k} \caK_{F_k}(a_k)  -  
        e^\ee  \mujw \caK_{\tau^k} \caK_{F_k}(a_k)  \DX(a_k)  \qd \prod_{\ell \neq k} 
        \underbrace{\int_A \mu \caK_{\tau^\ell} \caK_{F_\ell}(a_\ell) \DX{a_\ell} }_{=1} \\
        &\gla \DD^{F_k}_{\ \SAMP_{\tau^k}(\Djv),\ \SAMP_{\tau^k} (\Djw)}(\ee)
    \end{align*}}
    By combining this result with \eqref{equation:nonadaptive:spc:2} we get the bound:
    \KURZ{     \begin{IEEEeqnarray*}{rCl}
        &&\sum_{k=1}^{m} \frac{n_k}{n} \sum_{\tau \in \frC_{j,k}} \caT_{j,k}(\tau)
    \Max{S \subseteq A^m} \int_{S} \; \prod_{\ell \neq k} \mu \caK_{\tau^\ell}
  \caK_{F_\ell}(a_\ell) \\ 
        && \qd \left(  \mujv \caK_{\tau^k} \caK_{F_k}  -  e^\ee  
        \mujw \caK_{\tau^k} \caK_{F_k} \right)(a_k)
            \DX{\fra{}} \\
        &\gla& \sum_{k=1}^{m} \frac{n_k}{n} \sum_{\tau \in \frC_{j,k}} \caT_{j,k}(\tau) \  
        \DD^{F_k \circ \SAMP_{\tau^k}}_{\Djv,\Djw}(\ee) \\
        &\kla& \sum_{k=1}^{m} \frac{n_k}{n} \ \MSPP^{F_k, \caT^k_+}_{\caD, j} (\ee) \ . 
     \end{IEEEeqnarray*}
}
{    \begin{align*}
        &\sum_{k=1}^{m} \frac{n_k}{n} \sum_{\tau \in \frC_{j,k}} \caT_{j,k}(\tau) \Max{S \subseteq A^m} 
        \int_{S}  \left(  \mujv \caK_{\tau^k} \caK_{F_k}(a_k)  -  e^\ee  
        \mujw \caK_{\tau^k} \caK_{F_k}(a_k) \right)   
         \; \prod_{\ell \neq k} \mu \caK_{\tau^\ell} \caK_{F_\ell}(a_\ell)    \DX{\fra{}} \\
        &\gla \sum_{k=1}^{m} \frac{n_k}{n} \sum_{\tau \in \frC_{j,k}} \caT_{j,k}(\tau) \  
        \DD^{F_k}_{\ \SAMP_{\tau^k}(\Djv),\ \SAMP_{\tau^k} (\Djw)}(\ee) \\
   &  \kla \sum_{k=1}^{m} \frac{n_k}{n} \ \MSPP^{F_k, \caT^k_+}_{\caD, j} (\ee) \ . 
    \end{align*}
    }
\end{IEEEproof}

%When considering a d
For databases with a identical distributions of its entries,
the distribution of input values to answer a query are identical for all partitions 
except at the position where the critical entry is mapped to by the template.
Thus, in this case holds  
$\MSPP^{F_k,  \caT_{j,k}}_{\caD,j}(\ee) \gla   
\hDD^{F_k}_{\caD_{n_k},j}(\ee)$.
%\Rred{\\ $\DD_{F_k, \caD^{n_k}}(\ee) $ --   $\DD$ ist in Def. 2.3 anders definiert}
%atabase where the entires are identically independent distributed 
%the choice of entries for the different partitions does not change
%the densities of the distribution of the answer.

\begin{theorem}\label{theorem:composition:unadaptive:iid}
    Let $\mu$ be an iid distribution for a database $\caD$ of size $n$, $F$ a composition of
    $m$ queries $F_1, \dots, F_m$ and 
    $\T$  the injective $(n_1, \dots, n_m)$--\DRP\ of the entries.
%    $\T$ a $(n_1, \dots, n_m)$--splitting.
    Then $F \circ SAMP_{\T}$ achieves $(\ee, \; \sum_{k=1}^{m} \frac{n_k}{n} \;
\hat{\DD}^{F_k}_{\caD_{n_k},j}(\ee))$--statistical privacy for $\mu$ and $j$.
%    \DD_{F_k, \caD_{n_k}}(\ee) )
\end{theorem}
%\begin{IEEEproof}
%Since the entries are identically distributed it holds 
%\[ \MSPP_{F_k, \caD, \caT^k_+} \gla \DD_{F_k, \caD^{n_k}}(\ee). \]
%Then applying Lemma \ref{lemma:nonadaptive_composition:pure_spc} yields the theorem.
%\end{IEEEproof}

%\Rred{26.9.25. bis hierhin Notation korrigiert}

For nonidentical distributed entries, we obtain the following generalized version of the theorem.

\begin{theorem}{\rm \wfinw{Composition}}\label{theorem:composition:unadaptive:id}
    Let $\mu$ be a distribution for databases of size $n$, $F$ a composition of
    $m$ queries $F_1, \dots, F_m$ and $\T$ is a $(n_1, \dots, n_m)$--\DRP.
    Then $F \circ SAMP_{\T}$ with
\[ \dd \gla \sum_{k=1}^{m} \frac{n_k}{n} \;  \MSPP^{F_k, \caT^k_+}_{\caD, j} (\ee) \]
     achieves $(\ee, \dd)$--statistical privacy for $\mu$ and $j$.
\end{theorem}
%\begin{IEEEproof}
%    This is a direct consequence of Lemma \ref{lemma:nonadaptive_composition:pure_spc}.
%\end{IEEEproof}

%\Rred{Hier inhalte zur optimalen wahl von $n_1,\dots,n_m$!}

As a first approach, it may seem sensible to choose
the same size $\hat{n}$ for all partitions but since the privacy curves of the
different subqueries may vary strongly with respect to the database size a
sensible curator may want to allows for some subqueries to run on bigger
databases then others to maximize privacy. Furthermore for the sake of utility
the curator may want to grant any inquirer some control over the size of the
underlying databases to allow them for control over the error of the different
queries. In the following theorem, we examine the case where the curator of the
database, possibly in consultation with the inquirer, has opted for fixed
partition sizes which are known to the inquirer and what this means for privacy
in case of an malicious inquirer.

\subsection{Adaptive Composition}

For adaptive compositions the query can as well be modeled as a single query $F=
(F_1, \dots, F_m):\ W^n \mapsto A^m$. But in this case all subqueries depend on
the answer to the former queries meaning $F_j: A^{j-1} \times W^n \mapsto A$.
For a is a random partition $\T$ of the entries let
$\Tjk $ denote $\T $ conditioned on $j$ being used for the $k$-th query.

\begin{lemma}\label{lemma:adaptive_composition:pure_spc}
Let $\mu$ be a distribution on databases of size $n$, 
$F \gla (F_1, \dots, F_m)$ an adaptive composition of
$m$ queries \ and $\T$ a $(n_1, \dots, n_m)$--\DRP.
Then $F \circ SAMP_{\T}$  achieves $(\ee, \dd)$--statistical privacy for $\mu$ with
 \begin{IEEEeqnarray*}{rCl}
    \dd    &=& \sum_{k=1}^{m} \frac{n_k}{n} \sum_{\tau \in \frC_{j,k}}
        \caT_{j,k}(\tau) \ \int_{A^{k-1}}  \DD^{F_k(\fra{}_{k-1}) \circ \SAMP_{\tau^k}}_{\
        \Djv,\Djw}(\ee) \\
        && \prod_{\ell < k} \mu
        \caK_{\tau^\ell} \caK_{F_\ell(\fra{}_{\ell-1})}(a_\ell)
        \DX{\fra{}_{k-1}}  \ .
     \end{IEEEeqnarray*}
%        &\leq& 
%    \[ \dd =  \sum_{k=1}^{m} \frac{n_k}{n} \; \mathbf{E} \left[ \mathbf{E}_{X \sim
%    \mu \caK_{\hat{\caT}^{k-1}} \caK_{(F_1, \dots, F_{k-1})}} \left[
%%    \max_{j \in [1:n]} 
%    \MSPP^{F_k(X),  \bar{\caT}^k_+}_{\caD,,j} (\ee) \right] \right].
%    \]
\end{lemma}

\begin{IEEEproof}
    Let  $v,w \in W$ be arbitrarily chosen $\fra{} = (a_1, \dots,a_m)$.
As in the  nonadaptive proof we first split depending on the position of the sensitive element.
    \begin{IEEEeqnarray}{l}
        \DD^{F \circ \SAMP_\T}_{\Djv, \Djw}(\ee) \nonumber \\
        \kla  \sum_{k=1}^{m} \frac{n_k}{n} \ \Max{S \subseteq A^m} \nonumber \\
       \qd \qd \int_{S}   \mujv \caK_{\caT_{j,k}} \caK_F(\fra{}) -  e^\ee \mujw
        \caK_{\caT_{j,k}} \caK_F(\fra{})   \DX{\fra{}}
        \label{equation:adaptive:spc:1}
    \end{IEEEeqnarray}
    To handle the answers where the sensitive entry is not used one should first
    recognize that because of disjointedness of the subsets and the independence
    of the entries the subqueries are independent for any fixed $\tau$ and it
    holds 
    \[ 
        \mu \caK_{\tau} \caK_{F}(\fra{}) \gla \prod_{\ell=1}^{m} \mu \caK_{\tau^\ell} \caK_{F_\ell(\fra{}_{\ell-1})}(a_\ell). 
    \]
    Let $\fra{}_j = (a_1, \dots,
    a_j)$ for $0 \leq j \leq m$. For $\ell \neq k$ the sensitive element $j$ is not contained in the $l$-th
    subset of the random partition, thus $\mujv \caK_{\tau^\ell}
    \caK_{F_\ell(\fra{}_{\ell-1})}(a_\ell) \gla \mujw \caK_{\tau^\ell} \caK_{F_\ell(\fra{}_{\ell-1})}(a_\ell) \gla
    \mu \caK_{\tau^\ell} \caK_{F_\ell(\fra{}_{\ell-1})}(a_\ell)$. Thus
    \KURZ{ \begin{IEEEeqnarray*}{rCl} 
    &&\mujv \caK_{\caT_{j,k}} \caK_F(\fra{}) -  e^\ee \mujw \caK_{\caT_{j,k}}
    \caK_F(\fra{}) \\
        &\gla& \sum_{\tau \in \frC_{j,k}} \caT_{j,k}(\tau)  \left(\mujv \caK_{\tau} \caK_F(\fra{}) -  e^\ee \mujw \caK_{\tau} \caK_F(\fra{}) \right) \\
        &\gla& \sum_{\tau \in \frC_{j,k}} \caT_{j,k}(\tau)    \; \prod_{\ell
        \neq k} \mu \caK_{\tau^\ell} \caK_{F_\ell(\fra{}_{\ell-1})}(a_\ell)
        \\
        && \hspace*{5ex}\left(  \mujv \caK_{\tau^k} \caK_{F_k(\fra{}_{k-1})}  -  e^\ee
        \mujw \caK_{\tau^k} \caK_{F_k(\fra{}_{k-1})} \right)(a_k)
     \end{IEEEeqnarray*}}{    \begin{align*} 
        &\mujv \caK_{\caT_{j,k}} \caK_F(\fra{}) -  e^\ee \mujw \caK_{\caT_{j,k}} \caK_F(\fra{}) \gla \sum_{\tau \in \frC_{j,k}} \caT_{j,k}(\tau)  \left(\mujv \caK_{\tau} \caK_F(\fra{}) -  e^\ee \mujw \caK_{\tau} \caK_F(\fra{}) \right) \\
        &\gla \sum_{\tau \in \frC_{j,k}} \caT_{j,k}(\tau) \left(  \mujv \caK_{\tau^k} \caK_{F_k(\fra{}_{k-1})}(a_k)  -  e^\ee  \mujw \caK_{\tau^k} \caK_{F_k(\fra{}_{k-1})}(a_k) \right)    \; \prod_{\ell \neq k} \mu \caK_{\tau^\ell} \caK_{F_\ell(\fra{}_{\ell-1})}(a_\ell) 
    \end{align*}}
    \KURZ{    Let 
    \begin{IEEEeqnarray*}{rl}
        \eta_{\tau, w}(a_\ell) &\dea \mujw \caK_{\tau^\ell}
    \caK_{F_\ell(\fra{}_{\ell-1})}(a_\ell)\\
     \eta_{\tau}(a_\ell) &\dea \mu \caK_{\tau^\ell}
    \caK_{F_\ell(\fra{}_{\ell-1})}(a_\ell)
    \end{IEEEeqnarray*} 
    and
    \begin{IEEEeqnarray*}{l}
                B(a_k) := \left(\mujv
    \caK_{\tau^k}
    \caK_{F_k(\fra{}_{k-1})}  -  e^\ee  \mujw \caK_{\tau^k}
    \caK_{F_k(\fra{}_{k-1})}\right)(a_k)
    \end{IEEEeqnarray*}}{}
    Inserting this into the integral in \eqref{equation:adaptive:spc:1} it holds 
    \KURZ{ \begin{IEEEeqnarray}{l}
        \Max{S \subseteq A^m} \int_{S}  \mujv \caK_{\caT_{j,k}} \caK_F(\fra{}) -  e^\ee \mujw \caK_{\caT_{j,k}} \caK_F(\fra{}) \DX{\fra{}} \nonumber \\
        =\Max{S \subseteq A^m} \int_{S}  \sum_{\tau \in \frC_{j,k}} \caT_{j,k}(\tau) B(a_k)  \; \prod_{\ell \neq k} \eta_{\tau}(a_\ell) 
        \DX{\fra{}} \nonumber \\
        \leq \sum_{\tau \in \frC_{j,k}} \caT_{j,k}(\tau) \Max{S \subseteq
        A^m} \int_{S}  B(a_k)  \; \prod_{\ell \neq k} \eta_{\tau}(a_\ell)
        \DX{\fra{}} \nonumber
    \end{IEEEeqnarray}}{    \begin{align}
        &\sum_{i=k}^{m} \frac{n_k}{n} \Max{S \subseteq A^m} \int_{S}  \mujv \caK_{\caT_{j,k}} \caK_F(\fra{}) -  e^\ee \mujw \caK_{\caT_{j,k}} \caK_F(\fra{}) \DX{\fra{}} \nonumber \\
        &\gla  \sum_{k=1}^{m} \frac{n_k}{n} \Max{S \subseteq A^m} \int_{S}  \sum_{\tau \in \frC_{j,k}} \caT_{j,k}(\tau) \left(  \mujv \caK_{\tau^k} \caK_{F_k(\fra{}_{k-1})}(a_k)  -  e^\ee  \mujw \caK_{\tau^k} \caK_{F_k(\fra{}_{k-1})}(a_k) \right)    \; \prod_{\ell \neq k} \mu \caK_{\tau^\ell} \caK_{F_\ell(\fra{}_{\ell-1})}(a_\ell)  
        \DX{\fra{}} \nonumber \\
        &\kla \sum_{k=1}^{m} \frac{n_k}{n} \sum_{\tau \in \frC_{j,k}} \caT_{j,k}(\tau) \Max{S \subseteq A^m} \int_{S}  \left(  \mujv \caK_{\tau^k} \caK_{F_k(\fra{}_{k-1})}(a_k)  -  e^\ee  \mujw \caK_{\tau^k} \caK_{F_k(\fra{}_{k-1})}(a_k) \right)    \; \prod_{\ell \neq k} \mu \caK_{\tau^\ell} \caK_{F_\ell(\fra{}_{\ell-1})}(a_\ell)  
        \DX{\fra{}} 
        %\label{equation:adaptive:spc:2}
    \end{align}}
Again splitting the integral allows us
    to extract all densities independent of the sensitive entry.
  Then we use we use the fact the maximizing integral is over the set whole $A$.
    \KURZ{
         \begin{IEEEeqnarray*}{rCl}
        &&\Max{S \subseteq A^m } \int_{S} \left(\eta_{\tau, v}  -  e^\ee  \eta_{\tau, w} \right)(a_k) \; \prod_{\ell \neq k} \eta_{\tau}(a_\ell) \DX{\fra{}} \\
        &=& \Max{S \subseteq A^m } \int_{S_1} \dotsi \int_{S_m} \left(
        \eta_{\tau, v}  -  e^\ee  \eta_{\tau, w} \right)(a_k) \;
        \prod_{\ell \neq k} \eta_{\tau}(a_\ell) \\
        && \hspace*{5ex} \DX{a_m} \dotso \DX{a_1} \\
        &=& \Max{S \subseteq A^m } \int_{S_1} \dotsi \int_{S_k} \prod_{\ell < k}
        \eta_{\tau}(a_\ell) \; \left( \eta_{\tau, v}  -  e^\ee \eta_{\tau, w} \right)(a_k)  \\
        && \qd \int_{S_{k+1}} \dotsi \int_{S_{m-1}}  \prod_{k < \ell < m}
        \eta_{\tau}(a_\ell) \; \underbrace{\int_A \eta_{\tau}(a_m) \DX{a_m}
        }_{=1} \\
        && \hspace*{7ex}  \DX(a_{m-1}) \dotso
        \DX{a_1}  \\
        && \vdots \\
        &=& \Max{S \subseteq A^k } \int_{S_1} \dotsi \int_{S_k} \prod_{\ell < k}
        \eta_{\tau}(a_\ell) \; \left( \eta_{\tau, v}  -  e^\ee  \eta_{\tau,
        w} \right)(a_k) \\ && \hspace*{7ex} \DX(a_{k}) \dotso \DX{a_1}
     \end{IEEEeqnarray*}
    }{\begin{align*}
        &\Max{S \subseteq A^m } \int_{S} \left( \mujv \caK_{\tau^k} \caK_{F_k(\fra{}_{k-1})}(a_k)  -  e^\ee  \mujw \caK_{\tau^k} \caK_{F_k(\fra{}_{k-1})}(a_k) \right) \; \prod_{\ell \neq k} \mu \caK_{\tau^\ell} \caK_{F_\ell(\fra{}_{\ell-1})}(a_\ell) \DX{\fra{}} \\
        &\gla \Max{S \subseteq A^m } \int_{S_1} \dots \int_{S_m} \left( \mujv \caK_{\tau^k} \caK_{F_k(\fra{}_{k-1})}(a_k)  -  e^\ee  \mujw \caK_{\tau^k} \caK_{F_k(\fra{}_{k-1})}(a_k) \right) \; \prod_{\ell \neq k} \mu \caK_{\tau^\ell} \caK_{F_\ell(\fra{}_{\ell-1})}(a_\ell) \DX{a_m} \dots \DX{a_1} \\
        &\gla \Max{S \subseteq A^m } \int_{S_1} \dots \int_{S_k} \prod_{\ell < k} \left( \mu \caK_{\tau^\ell} \caK_{F_\ell(\fra{}_{\ell-1})}(a_\ell) \right) \; \mujv \caK_{\tau^k} \caK_{F_k(\fra{}_{k-1})}(a_k)  -  e^\ee  \mujw \caK_{\tau^k} \caK_{F_k(\fra{}_{k-1})}(a_k) \\
        &\qd \qd \int_{S_{k+1}} \dots \int_{S_{m-1}}  \prod_{k < \ell < m} \left( \mu \caK_{\tau^\ell} \caK_{F_\ell(\fra{}_{\ell-1})}(a_\ell) \right) \underbrace{\int_A \mu \caK_{\tau^m} \caK_{F_m(\fra{}_{m-1})}(a_m) \DX{a_m} }_{=1}  \DX(a_{m-1}) \dots \DX{a_1}  \\
        &\qd \vdots \\
        &\gla \Max{S \subseteq A^k } \int_{S_1} \dots \int_{S_k} \prod_{\ell < k} \left( \mu \caK_{\tau^\ell} \caK_{F_\ell(\fra{}_{\ell-1})}(a_\ell) \right) \; \mujv \caK_{\tau^k} \caK_{F_k(\fra{}_{k-1})}(a_k)  -  e^\ee  \mujw \caK_{\tau^k} \caK_{F_k(\fra{}_{k-1})}(a_k) \DX(a_{k}) \dots \DX{a_1}
    \end{align*}}
%    Since the queries previous to the one using the sensitive entry only
%    influence the previous term by their density the maximizing integral is over      $A$. 
%\Rred{
%    Using Fubini an additional time}
%and 
Using $\eta(a) \dea \prod_{\ell < k} \eta_{\tau}(a_\ell)$ the last term becomes
    \KURZ{     \begin{IEEEeqnarray*}{rCl}
    &&\Max{S \subseteq A } \int_{A^{k-1}} \int_{S_k}  \left( \eta_{\tau, v}   -
    e^\ee  \eta_{\tau, w}  \right)(a_k) \DX(a_{k}) \; \eta(a) \DX{\fra{}_{k-1}} \\
        &&\leq  \int_{A^{k-1}} \Max{S \subseteq A } \int_{S_k}  \left(
        \eta_{\tau, v}  -  e^\ee  \eta_{\tau, w} \right)(a_k) \DX(a_{k})
       \; \eta(a) \DX{\fra{}_{k-1}} \\
        &&\leq  \int_{A^{k-1}}     \DD^{F_k(\fra{}_{k-1})\circ
        \SAMP_{\tau^k}}_{\Djv, \Djw}(\ee) \; \eta(a) \DX{\fra{}_{k-1}} 
     \end{IEEEeqnarray*}}{ \begin{align*}
        &\Max{S \subseteq A } \int_{A^{k-1}} \int_{S_k}  \left( \mujv \caK_{\tau^k} \caK_{F_k(\fra{}_{k-1})}(a_k)  -  e^\ee  \mujw \caK_{\tau^k} \caK_{F_k(\fra{}_{k-1})}(a_k) \right) \DX(a_{k})   \prod_{\ell < k} \mu \caK_{\tau^\ell} \caK_{F_\ell(\fra{}_{\ell-1})}(a_\ell) \DX{\fra{}_{k-1}} \\
        &\kla  \int_{A^{k-1}} \Max{S \subseteq A } \int_{S_k}  \left( \mujv \caK_{\tau^k} \caK_{F_k(\fra{}_{k-1})}(a_k)  -  e^\ee  \mujw \caK_{\tau^k} \caK_{F_k(\fra{}_{k-1})}(a_k) \right) \DX(a_{k})   \prod_{\ell < k} \mu \caK_{\tau^\ell} \caK_{F_\ell(\fra{}_{\ell-1})}(a_\ell) \DX{\fra{}_{k-1}} \\
        &\kla  \int_{A^{k-1}}  \DD^{F_k(\fra{}_{k-1})}_{\ \SAMP_{\tau^k}(\Djv),\ \SAMP_{\tau^k} (\Djw)}(\ee) \prod_{\ell < k} \mu \caK_{\tau^\ell} \caK_{F_\ell(\fra{}_{\ell-1})}(a_\ell) \DX{\fra{}_{k-1}} 
    \end{align*}}
%    Let $\bar{\caT}^k = \caT^k \mid \hat{\caT}^{k-1}$. Thus bringing all  together it holds.
Finally,
     \begin{IEEEeqnarray*}{rCl}
        && \sum_{k=1}^{m} \frac{n_k}{n} \sum_{\tau \in \frC_{j,k}}
        \caT_{j,k}(\tau) \Max{S \subseteq A^m} \int_{S} \\
        &&  \left(  \mujv
        \caK_{\tau^k} \caK_{F_k(\fra{}_{k-1})}  -  e^\ee  \mujw
        \caK_{\tau^k} \caK_{F_k(\fra{}_{k-1})} \right)(a_k)  \\
        && \prod_{\ell \neq k} \mu \caK_{\tau^\ell} \caK_{F_\ell(\fra{}_{\ell-1})}(a_\ell)  
        \DX{\fra{}} \\
        &\leq& \sum_{k=1}^{m} \frac{n_k}{n} \sum_{\tau \in \frC_{j,k}}
        \caT_{j,k}(\tau) \ \int_{A^{k-1}}  \DD^{F_k(\fra{}_{k-1})\circ \SAMP_{\tau^k} }_{\
        \Djv,\Djw}(\ee) \\
        && \prod_{\ell < k} \mu
        \caK_{\tau^\ell} \caK_{F_\ell(\fra{}_{\ell-1})}(a_\ell)
        \DX{\fra{}_{k-1}}  \ .
     \end{IEEEeqnarray*}
\end{IEEEproof}

This implies

\begin{theorem}\label{theorem:composition:adaptive:iid}
Let $\mu$ be a distribution for an iid databases of size $n$, $F$ an adaptive composition of
  $m$ queries $F_1, \dots, F_m$ and $\T$  a corresponding $(n_1, \dots, n_m)$--\DRP.
%  with    $N_{k} = \sum^{k}_{h=1} n_h$.
Then $F \circ SAMP_{\T}$ achieves 
$(\ee,  \; \sum_{k=1}^{m} \frac{n_k}{n} \; \mathbf{E}_{X \sim (F_1, \dots, F_{k-1})}
    %(\caD^{N_{k-1}})} 
\left[ \DD^{F_k(X)}_{\caD^{n_k}}(\ee)  \right] )$--statistical privacy for $\mu$.
\end{theorem}
%\Rred{
%\begin{IEEEproof}
%    This is a result of Lemma \ref{lemma:adaptive_composition:pure_spc} and the
%    iid distribution where all entries of the databases are of the same
%    distribution such that the sampling does not affect the database
%    distribution. Thus the queries previous to the one which includes the
%    sensitive element are asked on a Database $\caD^{n^*}$ where $n^*$ is the
%    corresponding database sample size. 
%\end{IEEEproof}
%Omitting the identical distribution, we obtain the following generalized version of the theorem.
%}

If entries have different distributions an attacker can try to learn how these
distributions are partitioned among the queries.
Therefore, the choice of a query depends on 
$\prod_{\ell < k} \mu \caK_{\tau^\ell} \caK_{F_\ell(\fra{}_{\ell-1})}(a_\ell)$
which bounds the information he can get.
Thus we get

\begin{theorem}{\rm \wfinw{Composition}}\label{theorem:composition:adaptive:id}
    Let $\mu$ be a distribution for databases of size $n$, $F$ an adaptive composition of
    $m$ queries $F_1, \dots, F_m$ and $\T$ is a $(n_1, \dots, n_m)$--\DRP.
    Then $F \circ SAMP_{\T}$   achieves $(\ee, \dd)$--statistical privacy for 
\KURZ{
 \begin{IEEEeqnarray*}{rCl}
        \dd &=& \sum_{k=1}^{m} \frac{n_k}{n} \sum_{\tau \in \Cjk}
        \Tjk(\tau)  \int_{A^{k-1}}  \DD^{F_k(\fra{}_{k-1}) \circ
        \SAMP_{\tau^k}}_{\Djv,\Djw}(\ee) \\
        && \hspace*{3ex} \prod_{\ell < k} \mu \caK_{\tau^\ell} \caK_{F_\ell(\fra{}_{\ell-1})}(a_\ell) \DX{\fra{}_{k-1}} \ .
     \end{IEEEeqnarray*}
}
{    \[ \dd = \sum_{k=1}^{m} \frac{n_k}{n} \sum_{\tau \in \frC_{j,k}}
    \caT_{j,k}(\tau)  \int_{A^{k-1}}  \DD^{F_k(\fra{}_{k-1}) \circ \SAMP_{\tau^k}}_{\Djv,\Djw}(\ee) \prod_{\ell < k} \mu
    \caK_{\tau^\ell} \caK_{F_\ell(\fra{}_{\ell-1})}(a_\ell) \DX{\fra{}_{k-1}}
    \]
}
\end{theorem}
%\begin{IEEEproof}
%    This is a direct consequence of Lemma \ref{lemma:adaptive_composition:pure_spc}.
%\end{IEEEproof}
%

\section{Applications}

Good privacy guarantees are useless if the distortion of the data, the utility loss
$\caU$, is large, resp.~the accuracy is small.
$\caU$ is typically be measured by the mean square error (MSE) introduced.
This is demonstrated by the mean estimator where the MSE falls
linearly with the database size. As a general rule for poisson sampling and
sampling without replacement for an iid database
%, which is
%a fundamental assumption when using any estimator since they are used to infer
%the parameters of the underlying database, 
any subsampling will result in an
accuracy loss that stems from the smaller dataset. 
%When
%considering sampling with replacement the influence on the accuracy varies
%strongly depending on the underlying distributions.

To compare different settings and mechanisms an obvious way is to fix two
of the parameters $\ee,\dd, \caU$ to some reasonable values
and then to calculate how the remaining one behaves.

%When discussing the utility of queries and privacy mechanisms a widely accepted
%measure for accuracy is the mean square error (MSE) of the mechanism and query
%combination. Since the accuracy depends of the size of the database subsampling
%can be a problem. 

\begin{example}
    For a property queriy $F$ with property probability $p$
    and sampling without replacement $\Tns$ the  MSE increases by
    \begin{IEEEeqnarray*}{C}
         \frac{p (1-p)}{s} - \frac{p (1-p)}{n}  
   \end{IEEEeqnarray*}
    since the number of positive elements is binomially distributed with
    parameters $n$ respective $s$ and $p$.
\end{example}

The way in which subsampling or resampling
methods effect the accuracy of different estimators is subject of the resampling
field in statistics.
Typical  applications of sampling methods considered here are
Bootstrapping, Cross validation or Jackknife (for an overview of these methods see \cite{G99}).

% ICH HABE MICH HIER VERRANNT MORGEN NOCHMAL DRAUF SCHAUEN 

% \begin{example}
%     To still illustrate this case we consider an estimator $F$ on a database of
%     size $n$ that behaves $\mathcal{N}(\mu, \sigma^2/n)$  then how much noise is
%     still needed to achieve $(\ee = 0.01, \dd = 1/n)$--SP. Further we assume
%     that the sensitivity meaning the biggest possible difference in the mean of
%     the estimator depending on the value of one entry is $\Delta F = 1/n$. Then
%     for $n = 200$ we calculate the additionally needed noise for different
%     sample sizes by numerically approximating the 
% \end{example}

%\subsection{Missing privacy in composition}

Composition has extensively been investigated for DP.
However, large values for $\ee$ are used that have little practical use.
For small $\ee$ on the other side the accuracy often is too low if
noise necessary according to the DP bounds are injected.
Is  the SP scenario suitable to overcome these problems? 
%and bring PET into practice.
% is a very competitive field where every possible way to
%further reduce noise is used and still many use privacy parameters of $\ee > 5$
%or worse and still seem to uphold some real privacy. 
%To illustrate how much of this may be lead back to the uncertainty of an adversary 
To compare SP with DP we have calculated the number of property queries we could
ask in DP compared to the amount of property queries we can ask over an database
with iid distributed queries in SP for the same privacy parameters and accuracy
using a Gaussian Noise mechanism for DP calibrated using the noise formula in
\cite{KOV15}.  

% We did this by fixing $\ee$ for the composed query and calculated
% the $\dd$ parameter for SP. Then we calculated the loss of
% accuracy that occurs for the subqueries in the SP model. Know we used Theorem 15
% in \cite{KOV15} to calculate the number of queries that could be asked for
% $(\ee,\dd)$ such that the Gaussian noise added to the subqueries amounted to the
% same accuracy loss as the queries in the SP model obtained. 

In more detail, first $\ee$ is fixed for a composed property query with \#SP
subqueries and the corresponding $\dd$ parameter according to SP is computed
using Theorem~\ref{theorem:composition:unadaptive:iid} and an idd database of
size $n$. Then the accuracy loss incurred by the subqueries in the SP model is
calculated. Finally, using Theorem~15 in \cite{KOV15} it is estimated how many
queries \#DP can be asked under $(\ee, \dd)$--DP such that the Gaussian noise
added to the subqueries leads to the same accuracy loss as observed in the SP
model. 

% Considering the results of these calculations in the tables
% \ref{table:NR_DP_Querys_Compared_to_SP_Queries,table:NR_DP_Querys_Compared_to_SP_Queries_small_n}
% there is enough internal noise to often ask multiple queries on a partitioned database lowering the
% accuracy then using a Gaussian Noise mechanism to protect the privacy. 
% This illustrates one of the key problems when utilizing DP in real world
% application -- namely that the parameters used often do not fit the privacy
% gained. Here DP would only allow us a few queries under a reasonable low $\ee$
% enforcing the use of high $\ee$ parameters when we can see that much of the
% actual privacy is generated by internal noise. 

\begin{table}[!t]
\renewcommand{\arraystretch}{1.3}
\caption{Results when using a database of size $n = 2^{15} = 32768$
and occurrence probability of $p=1/2$ for the property queries.}
\label{table:NR_DP_Querys_Compared_to_SP_Queries}
\centering
\begin{tabular}{ | c | c | c | c | c | c | c | c |}
    \multicolumn{2}{c}{} & \multicolumn{2}{c}{$\ee = 0.005$} &
    \multicolumn{2}{c}{$\ee = 0.01$} &
    \multicolumn{2}{c}{$\ee = 0.02$} \\
    \hline
    \#SP & $\sigma$ & $\dd$ & \#DP & $\dd$ & \#DP & $\dd$ & \#DP  \\
    \hline
    32 & .0153& .0225 & 0  & .0203   &  3 & .0163 & 9\\
    \hline
    64  & .0219 & .0329 &  1  & .0306 &  6 & .0264 & 20  \\
    \hline
    128  & .0311  & .0475 &  3  & .0452 &  12 & .0409 & 44 \\
    \hline
    256  & .0441  & .0682 &  6 & .0660  &  25  & .0617 & 93\\
    \hline
    512  & .0624  & .0973 &  12  & .0953 & 50 & .0912 & 194 \\
    \hline 
\end{tabular}
\end{table}

\begin{table}[!t]
\renewcommand{\arraystretch}{1.3}
\caption{Results when using a database of size $n = 2^{10} = 1024$
and occurrence probability of $p=1/2$ for the property queries.}
\label{table:NR_DP_Querys_Compared_to_SP_Queries_small_n}
\centering
\begin{tabular}{ | c | c | c | c | c | c | c | c |}
    \multicolumn{2}{c}{} & \multicolumn{2}{c}{$\ee = 0.05$} &
    \multicolumn{2}{c}{$\ee = 0.1$} &
    \multicolumn{2}{c}{$\ee = 0.2$} \\
    \hline
    \#SP & $\sigma$ & $\dd$ & \#DP & $\dd$ & \#DP & $\dd$ & \#DP  \\
    \hline
    32 & .0869& .1214 & 2  & .1020   &  7 & .0.0711 & 23\\
    \hline
    64  & .1240 & .1808 &  4  & .1644 &  16 & .1291 & 55  \\
    \hline
    128  & .1760  & .2618 &  9 & .2496 &  35 & .2232 & 126 \\
    \hline
\end{tabular}
\end{table}
\Rblue{
For strong $\ee$ regimes we observe a substantial difference in the number of queries that
can be issued depending on if the internal noise is taken into account (see
Tables
\ref{table:NR_DP_Querys_Compared_to_SP_Queries}, \ref{table:NR_DP_Querys_Compared_to_SP_Queries_small_n}).
Although this effect is less pronounced for more relaxed $\ee$ parameters, it
remains particularly relevant for smaller databases, where even at comparatively
large $\ee = 0.1$ we still see a significant gap.}

\section{Conclusion}
% Resampling methods serve an important role in statistics where privacy is of
% high concern. Since many researchers still reject additive noise mechanisms we
% showed that sampling methods can be applied to amplify privacy provided by
% internal noise to reach acceptable levels. After handling the case of a single
% query we have extended the investigation to the composition of queries. Here we
% provide a method to calculate composed queries in the SP setting with stronger
% bounds that the former ones \cite{BBG11} and developed tools to handel different
% kinds of further background noise or non identically distributed databases
% allowing for an easier implementation of SP for single and compose queries. We
% compared the composition results to DP where we observe that the internal noise
% is often sufficient to allow multiple queries on a partitioned database while
% still degrading accuracy less than a Gaussian noise mechanism would need to to
% ask the same amount of queries. This highlights a central issue in applying DP
% in real-world scenarios: the chosen parameters frequently do not reflect the
% actual privacy obtained. Under differential privacy, only a small number of
% queries would be permitted for reasonably low $\ee$, effectively forcing
% practitioners to use high $\ee$ values -- even though a substantial portion of
% the privacy protection is already provided by internal noise.

\Rblue{
In the statistical privacy setting, we have introduced a method for deriving
composition bounds in the SP framework that were strictly sharper than
previously established results \cite{BBG11}. Furthermore, we have developed tools
capable of accommodating additional forms of background noise as well as
non-identically distributed databases, thereby facilitating the practical
deployment of SP for both single and composed queries.}

\Rblue{A comparison with differential privacy (DP) reveals that internal noise is
frequently sufficient to permit multiple queries on a partitioned database while
incurring less accuracy degradation than would be required under a Gaussian
mechanism to support the same number of queries. This observation underscores a
fundamental limitation in the application of DP to real-world scenarios: the
prescribed parameter choices often do not align with the effective privacy
actually achieved in practice. Under DP, only a small number of queries is
admissible for reasonably small values of $\ee$, effectively compelling
practitioners to adopt substantially larger $\ee$ values -- even though a
considerable portion of the resulting privacy protection is already provided by
internal noise.}

\bibliographystyle{IEEEtran}
\bibliography{subsam}

\end{document}